\documentclass[9pt]{extarticle}
\usepackage[utf8]{inputenc}
\usepackage{amsmath,graphicx,mlspconf}
\usepackage{amsfonts}
\usepackage{amssymb}
\usepackage{amsthm}
\usepackage{float}
\usepackage[linkcolor=mar,colorlinks=true]{hyperref}
\hypersetup{colorlinks,breaklinks,
            linkcolor=[rgb]{0.5,0,0},
            citecolor=[rgb]{0.1484375,0.1484375,0.53125}}
\usepackage{cleveref}
\usepackage{cite}
\usepackage{enumitem}
\usepackage[ruled,vlined,linesnumbered]{algorithm2e}
\usepackage{algpseudocode}
\usepackage[skip=0pt]{caption}
\usepackage{bm}
\newtheorem{theorem}{Theorem}
\newtheoremstyle{reduced}
  {6pt} 
  {6pt} 
  {} 
  {} 
  {\bfseries} 
  {.} 
  {.5em} 
  {} 
\newtheorem{lemma}[theorem]{Lemma}
\newtheorem{proposition}[theorem]{Proposition}
\theoremstyle{reduced}
\newtheorem{assumption}{Assumption}

\makeatletter
\renewenvironment{proof}[1][\proofname]{\par
  \vspace{-\topsep}
  \pushQED{\qed}%
  \normalfont
  \topsep3pt \partopsep3pt 
  \trivlist
  \item[\hskip\labelsep
        \itshape
    #1\@addpunct{.}]\ignorespaces
}{%
  \popQED\endtrivlist\@endpefalse
  \addvspace{2pt plus 0pt} 
}
\makeatother

\renewenvironment{thebibliography}[1]{%
   \begin{oldthebibliography}{#1}%
     \setlength{\itemsep}{0ex}%
}%
{%
   \end{oldthebibliography}%
}

\allowdisplaybreaks

\usepackage{etoolbox}

\makeatletter
\patchcmd{\@algocf@start}
  {-1.5em}
  {0pt}
  {}{}
\makeatother

\copyrightnotice{978-1-7281-6338-3/21/\$31.00 {\copyright}2021 IEEE}

\toappear{2021 IEEE International Workshop on Machine Learning for Signal Processing, Oct.\ 25--28, 2021, Gold Coast, Australia}

\title{Model-Free Learning of Optimal Deterministic Resource Allocations in Wireless Systems via Action-Space Exploration}
\name{\normalsize  Hassaan Hashmi and Dionysios S. Kalogerias}
\address{\normalsize Department of Electrical Engineering, Yale University, New Haven, USA\\
\normalsize \tt \{hassaan.hashmi,dionysis.kalogerias\}@yale.edu}

\begin{document}
\setlist[itemize]{noitemsep, topsep=3pt}
\setlist[itemize]{leftmargin=*}

\maketitle

\begin{abstract}
\vspace{-1bp}
Wireless systems resource allocation refers to perpetual and challenging nonconvex constrained optimization tasks, which are especially timely in modern communications and networking setups involving multiple users with heterogeneous objectives and imprecise or even unknown models and/or channel statistics. In this paper, we propose a technically grounded and scalable primal-dual deterministic policy gradient method for efficiently learning optimal parameterized resource allocation policies. Our method not only efficiently exploits gradient availability of popular universal policy representations, such as deep neural networks, but is also truly model-free, as it relies on consistent zeroth-order gradient approximations of the associated random network services constructed via low-dimensional perturbations in action space, thus fully bypassing any dependence on critics. Both theory and numerical simulations confirm the efficacy and applicability of the proposed approach, as well as its superiority over the current state of the art in terms of both achieving near-optimal performance and scalability.
\end{abstract}
\vspace{-4bp}
\vspace{-8.5bp}
\section{Introduction}
\label{sec:intro}
\vspace{-7.5bp}
Optimally allocating resources in modern wireless systems presents three major challenges \cite{ribeiro2012optimal, kalogerias2020model}: \textit{Infinite dimensionality of policies}; \textit{nonconvex constraints}; \textit{system and channel model unavailability}. Building on now classical dual-domain techniques for globally optimal model-based resource allocation \cite{yu2006dual, ribeiro2012optimal}, as well as effective heuristics
\cite{shi2011iteratively,wu2013flashlinq},
\textit{machine learning for wireless communications} has developed as a rapidly expanding area since, aside from nonconvexity (which results in natural algorithmic limitations), it can effectively address both challenges of infinite dimensionality and (partially) model availability; see, e.g., \cite{sun2018learning, lee2019deep, nasir2019multi, kalogerias2020model, eisen2019learning, Meng2020, khan2020centralized}.

Reinforcement Learning (RL), and in particular \textit{policy gradient algorithms} are attractive for wireless systems resource allocation, because they can naturally accommodate continuous channel responses and allocation decisions, and can produce parameterized optimal resource allocation policies that operate in online, even adaptive settings. Methods involving deterministic and stochastic off-policy learning have been developed recently  \cite{Meng2020, wang2020drl}. These methods employ some form of a gradient function approximator (i.e., a critic), making them essentially model-based, since the choice of such an approximator depends on the structure/nature of the underlying problem.  A model-free primal-dual learning method was presented in \cite{eisen2019learning}, based on the standard stochastic policy gradient method. However, the method in \cite{eisen2019learning} injects noise both in training and implementation phases due to the use of stochastic policies, which is undesirable in the resource allocation setting. Moreover, although one-point policy gradient estimates might be improved with the use of baselines, the latter are computationally expensive, sample inefficient, and cannot be implemented in an online fashion. Recently, \cite{kalogerias2020model} introduced two-point zeroth-order policy gradient representations, optimizing purely deterministic resource allocation policies. However, similar to \cite{eisen2019learning}, the approach in \cite{kalogerias2020model} presents scalability issues, as the corresponding zeroth-order gradient representations  are evaluated directly in the policy parameter space, which can be prohibitively large.

To tackle the limitations in terms of applicability and scalability, we propose a new primal-dual zeroth-order deterministic policy gradient method (referred to as PD-ZDPG+), which relies on consistent zeroth-order two-point gradient approximations of the associated random network services, constructed via \textit{low-dimensional perturbations in action space}, rather than in parameter space (as in \cite{kalogerias2020model}). The advantages of our approach are \textit{three-fold}, as follows:
\vspace{-2bp}
\begin{itemize}
    \item First, as the policies are deterministic, they directly conform with the standard formulation of resource allocation problems in wireless settings, in which randomized policies are undesirable and potentially redundant (this is not achieved in \cite{eisen2019learning}, but it is the case in actor-critic-based methods, such as those in \cite{Meng2020, wang2020drl}).
    \vspace{-6bp}
    \item Second, our algorithm is completely model-free, as it depends only on direct service function evaluations (obtainable during the operation of the system), thus effectively bypassing dependence on critics (this is obviously not the case in actor-critic-based methods, such as those in \cite{Meng2020, wang2020drl}, but is naturally achieved in \cite{eisen2019learning} and \cite{kalogerias2020model}).
    \vspace{-6bp}
    \item Third, the complexity of action-space exploration is independent of the dimensionality of the employed policy parameterization and thus, in conjunction with availability of the gradient of the latter (in any off-the-shelf machine learning software library), our approach scales substantially better with respect to the complexity of the problem, compared with the state of the art (this not the case in \cite{kalogerias2020model}, but is achieved in \cite{eisen2019learning}).
\end{itemize}
\vspace{-3bp}
Consequently, the proposed method evidently combines the best elements among competing approaches in the current literature.

In terms of related work, our primal-dual method resembles general purpose zeroth-order deterministic policy gradient methods in multi-stage unconstrained RL \cite{kumar2020zeroth,vemula2019}; in this work, we essentially consider a single-stage RL formulation, however within a constrained learning setting, leading to a challenging minimax problem formulation. 

\sloppy Lastly, we evaluate the performance of the proposed method against competing policy gradient methods in \cite{eisen2019learning,kalogerias2020model, lillicrap2015continuous, Meng2020, wang2020drl}, adapted to the setting studied herein, and with respect to the points made above. Source code for our experiments may be found at: \href{https://github.com/hassaanhashmi/pd\_zdpg\_plus}{https://github.com/hassaanhashmi/pd\_zdpg\_plus}. Certain proofs are omitted due to lack of space, to be presented in a follow-up work.
\section{Problem Formulation}
\label{sec:prob}
We consider a generic wireless systems resource allocation problem of the form \cite{ribeiro2012optimal, eisen2019learning, kalogerias2020model}
\begin{align} 
\begin{aligned} \label{eq:1}
    \underset{\textbf{x},\bm{\theta}}{\text{maximize}} & \quad g^o(\mathbf{x}) \\
    \text{subject to} & \quad \mathbf{x} \leq \mathbb{E}\left\{ \mathbf{f}(\bm{\phi}(\mathbf{H},\bm{\theta}), \mathbf{H})\right\} \\
    & \quad \mathbf{g}(\mathbf{x})\geq\mathbf{0} , (\mathbf{x}, \bm{\theta}) \in \mathcal{X} \times \Theta
\end{aligned},
\end{align}
where $\mathbf{H} \in \mathcal{H} \subseteq \mathbb{R}^{N_H}$ are random fading channels of the network with a joint distribution $\mathcal{M}_\mathbf{H}$, $\mathbf{f}(\bm{\phi}(\mathbf{H},\bm{\theta}), \mathbf{H})$ are the instantaneous service level metrics, which are functions of the channel vector $\mathbf{H}$ and a parameterized and differentiable policy funtion $\bm{\phi}: \Theta \times \mathcal{H} \to \mathcal{A}\subseteq \mathbb{R}^{N_R}$, with $\mathcal{A}$ closed,
and $g^o:\mathcal{X} \to \mathbb{R}$ and $\mathbf{g}:\mathcal{X} \to \mathbb{R}^{N_g}$ are concave utility functions, usually chosen by the designer.
The expected value of $\mathbf{f}(\bm{\phi}(\mathbf{H},\bm{\theta}), \mathbf{H})$ bounds the ergodic (in a sense, long-term) service level metrics $\mathbf{x}\in \mathcal{X} \subseteq \mathbb{R}^{N_S}$.

Any reasonable iterative optimization algorithm for directly solving \eqref{eq:1} requires evaluations of the gradients $\nabla g^o(\mathbf{x})$, $\nabla \mathbf{g}(\mathbf{x})$ and $\nabla\,\mathbb{E}\{\mathbf{f}(\bm{\phi}(\mathbf{H},\bm{\theta}),\mathbf{H})\}$. As the exact form of the services $\mathbf{f}$ and the channel distribution $\mathcal{M}_{\mathbf{H}}$ are most often not known a priori (incomplete model knowledge due to dependence on propagation physics and increasing complexity of interference and multiple access management models), it is not possible to evaluate the gradients of these functions \cite{eisen2019learning, kalogerias2020model}. To address this limitation, we focus on a \textit{model-free setting}, i.e., we develop a stochastic approximation approach (bypassing the need for knowledge of $\mathcal{M}_{\mathbf{H}}$), in which we are able to approximate the above gradients using appropriate function evaluations, which may be obtained by merely probing the actual wireless system during policy training (also bypassing the need for gradient evaluations --and thus model information--).
\section{Smoothed Surrogates in Action Space}
\label{sec:surrogate}
We first introduce Gaussian-smoothed versions of the functions involved in problem \eqref{eq:1}. Such functions demonstrate, under some feasible assumptions, certain desirable properties, particularly pertaining to their gradient evaluations, which we use to develop the proposed method.
\subsection{Gaussian Smoothed Functions}
\label{subsection:3.1}
To begin, let us define
\begin{align}
\label{eq:7}
g^o_{\mu_S}(\mathbf{x}) &\triangleq \mathbb{E}\left\{ g^o (\Pi_{\mathcal{X}}\{\mathbf{x} + {\mu_S}{\bm{U}_S}\})\right\}, \: \mathbf{x} \in \mathcal{X},\\
\label{eq:8}
\mathbf{g}_{\mu_S}(\mathbf{x}) & \triangleq \mathbb{E}\left\{ \mathbf{g} (\Pi_{\mathcal{X}}\{\mathbf{x} + {\mu_S}{\bm{U}_S}\})\right\}, \: \mathbf{x} \in \mathcal{X} \:\: \text{and} \\
\label{eq:9}
\bar{\mathbf{f}}^\phi_{\mu_R}(\bm{\theta}) & \triangleq  \mathbb{E}\left\{\mathbf{f}(\Pi_\mathcal{A}\left\{ \bm{\phi}(\mathbf{H}, \bm{\theta}) + {\mu_R}{\bm{U}_R}\right\}, \mathbf{H})\right\},   \bm{\theta} \in \Theta,
\end{align}
where $\mu_S \geq 0$ and $\mu_R \geq 0$ are smoothing parameters, $\bm{U}_S \sim \mathcal{N}(\mathbf{0}, \mathbf{I}_S)$ and $\bm{U}_R \sim \mathcal{N}(\mathbf{0}, \mathbf{I}_R)$ are random Gaussian vectors, and $\Pi_{\mathcal{Y}}\{\cdot\}$ denotes projection onto a closed set $\mathcal{Y}$.  Also, we assume $\bm{U}_R$ and $\mathbf{H}$ to be independent. This formulation is new and interesting in that we are \textit{smoothing the function $\mathbf{f}$ in the action space $\mathcal{A}$}, rather than the parameter space $\Theta$, as was previously done in \cite{kalogerias2020model}. Then, we may introduce a \textit{smoothed surrogate to problem} \eqref{eq:1}, defined as
\begin{align}
\begin{aligned} \label{eq:10}
\underset{\textbf{x},\bm{\theta}}{\text{maximize}} & \quad g^o_{\mu_S}(\mathbf{x}) \\
\text{subject to} & \quad \mathbf{x} + \texttt{S}(\mu_R) \leq \bar{\mathbf{f}}^\phi_{\mu_R}(\bm{\theta}) \\
& \quad \mathbf{g}_{\mu_S}(\mathbf{x})\geq\mathbf{0} , (\mathbf{x}, \bm{\theta}) \in \mathcal{X} \times \Theta
\end{aligned},
\end{align}
where the non-negative feasibility slack $\texttt{S} : \mathbb{R}_+ \rightarrow \mathbb{R}_+^{N_S}$ prevents constraint violations in \eqref{eq:10} relative to \eqref{eq:1} and is assumed to have certain properties, as explained later (see Assumption \ref{assumption2}). Here, we note that functions $g^o$ and $\mathbf{g}$ are concave. Hence, their smoothed surrogates $g^o_{\mu_S}$ and $\mathbf{g}_{\mu_S}$ are their underestimators respectively \cite{nesterov2017random}, and thus do not require feasibility slacks.

Formulation of the surrogate \eqref{eq:10} enables zeroth-order gradient representations of $g_{\mu}^o$, $\mathbf{g}_{\mu_S}$ and $\bar{\mathbf{f}}^\phi_{\mu_R}$ (i.e., based only on corresponding function evaluations), whenever they appropriately demonstrate certain properties. In the following subsections, we formally define those properties for the surrogate functions and ascertain the feasibility of the surrogate program \eqref{eq:10}, in relation to the original resource allocation program \eqref{eq:1}.
\subsection{Properties of Surrogate Functions}
\label{Zeroth}
We impose appropriate structures on the $i$-th entries $g^i$ and $g^i_{\mu_S}$ of $\mathbf{g}$ and $\mathbf{g}_{\mu_S}$ where $i \in \mathbb{N}^+_{N_g}$, and on the entries $f^{i}$ and $\bar{f}^{\bm{\phi},i}_{\mu_R}$ of $\mathbf{f}$ and $\bar{\mathbf{f}}^{\bm{\phi}}_{\mu_R}$ where $i \in \mathbb{N}^+_{N_S}$, respectively.

\begin{assumption} \label{assumption1}
The following conditions are satisfied: 
\\
\textbf{\textup{C1}} For every $i \in \{o,\mathbb{N}^+_{N_\mathbf{g}}\}$, $g^i(\mathbf{x})$ is globally Lipschitz with constant $L^i_g < \infty$.
\\
\textbf{\textup{C2}} For every $i \in \mathbb{N}^+_{N_S}$, $f^i(\cdot, \mathbf{H})$ is Lipschitz on $\mathcal{A}$ with constant $L^i_f(\mathbf{H})$, such that $\|L^i_f(\mathbf{H})\|_{\mathcal{L}_2}< \infty$.
\\
\sloppy \textbf{\textup{C3}} For every $\bm{\theta}$ and 
almost every $\mathbf{H}$, the parameterization $\bm{\phi}(\mathbf{H},\cdot)$ is Lipschitz in a neighborhood of $\bm{\theta}$ with constant $L^{\bm{\theta}}_{\bm{\phi}}(\mathbf{H})$ such that $\|L^{\bm{\theta}}_{\bm{\phi}}(\mathbf{H})\|_{\mathcal{L}_2}< \infty$.
\end{assumption}

Condition \textbf{C2} of Assumption \ref{assumption1} has the following consequences on the behavior of $\mathbb{E}\{f^i(\cdot, \mathbf{H})\}$ for every $i \in \mathbb{N}^+_{N_S}$.
\begin{proposition} \label{proposition1}         
Let condition \textbf{\textup{C2}} of Assumption \ref{assumption1} be in effect. Then, for every $i \in \mathbb{N}^+_{N_S}$, $\mathbb{E}\{f^i(\cdot, \mathbf{H})\}$ is $\mathbb{E}\{L^i_f(\mathbf{H})\}$-Lipschitz on $\mathcal{A}$. Moreover, it is true that, for every $(\bm{\theta},\bm{u}) \in \Theta \times \mathbb{R}^{N_R}$,
\begin{flalign} \label{eq:11}
    &\quad\quad\quad\mathbb{E}\{ \lvert f^i(\Pi_\mathcal{A}\left\{\bm{\phi}(\mathbf{H},\bm{\theta}) {+} \bm{u}\right\}, \mathbf{H})\rvert \}  &\nonumber \\
    &\hspace{0.65in}\leq \mathbb{E}\{L^i_f(\mathbf{H})\} \| \bm{u} \|_2 {+} \mathbb{E}\{ \lvert f^i(\bm{\phi}(\mathbf{H},\bm{\theta}), \mathbf{H})\rvert \}.
\end{flalign}
\end{proposition}
\begin{proof}
We start with condition \textbf{C2} and employ Jensen's inequality and the triangle inequality, respectively.
\end{proof}
We now use Assumption \ref{assumption1} and Proposition \ref{proposition1} to establish well-definiteness and basic properties of $g^o_{\mu_S}$, $\mathbf{g}_{\mu_S}$ and $\bar{\mathbf{f}}^{\bm{\phi}}_{\mu_R}$: For $\mathbf{x} \in \mathcal{X}$, $\mu_S > 0$ and for every $i \in \{ o, \mathbb{N}^+_{N_\mathbf{g}}\}$, let us define finite differences
\begin{align} \label{eq:12}
\Delta^i_g(\mathbf{x},\mu_S,\bm{U}_S) \triangleq \frac{g^i(\Pi_{\mathcal{X}}\{\mathbf{x} + {\mu_S}{\bm{U}_S}\}) {-} g^i(\mathbf{x})}{\mu_S}.
\end{align}
Similarly, for all $\bm{\theta} \in \Theta$, $\mu_R > 0$, and for every $i \in \mathbb{N}^+_{N_S}$ we define
\begin{align} \label{eq:13}
&\Delta^i_f(\bm{\theta},\mu_R,\bm{U}_R,\mathbf{H})  \nonumber\\ 
&\triangleq \frac{f^i({\Pi}_\mathcal{A}\left\{\bm{\phi}(\mathbf{H},\bm{\theta}) {+} \mu_R \bm{U}_R \right\}, \mathbf{H}) - f^i(\bm{\phi}(\mathbf{H},\bm{\theta}), \mathbf{H})}{\mu_R}.
\end{align}

Using Assumption \ref{assumption1} and Proposition \ref{proposition1}, we now formally define the set of properties which enable us to evaluate zeroth-order gradients of the smoothed functions defined in Subsection \ref{subsection:3.1}.

\begin{lemma}
\label{lemma2}
Let condition \textbf{C1} of Assumption \ref{assumption1} be in effect. Then, for every $i \in \{o,\mathbb{N}^+_{N_\mathbf{g}}\}$ and for every $\mu_S > 0$, each $g^i_{\mu_S}$ is a well-defined, finite, concave and everywhere differentiable underestimator of $g^i$ on $\mathcal{X}$, such that
\begin{align}
\label{eq:14}
\sup_{\mathbf{x}\in \mathcal{X}} \lvert g^i_{\mu_S}(\mathbf{x}) - g^i(\mathbf{x})\rvert & \leq \mu_S L^i_g \sqrt{N_S},\\
\label{eq:15}
\sup_{\mathbf{x}\in \mathcal{X}} \mathbb{E}\{\| \Delta^i_g(\mathbf{x}, \mu_S, \bm{U}_S)\bm{U}_S \|^2_2\} & \leq ({L^i_g})^2(N_S + 4)^2,\\
\label{eq:18}
\text{and} \quad \mathbb{E}\{ \Delta^i_g(\mathbf{x}, \mu_S, \bm{U}_S)\bm{U}_S\} & \equiv \nabla g^i_{\mu_S}(\mathbf{x})
\end{align}
for all $\mathbf{x} \in \mathcal{X}$.
\end{lemma}
\begin{lemma}
\label{lemma3}
\sloppy Let condition \textbf{C2} of Assumption \ref{assumption1} be in effect. Then, for every $\mu_R > 0$ and for all $i \in \mathbb{N}^+_{N_S}$, each $\bar{f}^{\bm{\phi},i}_{\mu_R}$ and each $\mathbb{E}\{f^i(\bm{\phi}(\mathbf{H}, \cdot),\mathbf{H})\}{\equiv} \bar{f}^{\bm{\phi},i}(\cdot)$ are such that
\begin{align}
   \label{eq:19}
   \sup_{\bm{\theta}\in \Theta} 
   \lvert 
   \bar{f}^{\bm{\phi},i}_{\mu_R}(\bm{\theta}) {-} \bar{f}^{\bm{\phi},i}(\bm{\theta}) 
   \rvert 
   &
   \hspace{-1.5bp}\leq\hspace{-1bp} 
   \mu_R \mathbb{E}\{L^i_f(\mathbf{H})\} \sqrt{N_R}
   \\
   \label{eq:20}
   \hspace{-6bp}
   \sup_{\bm{\theta}\in \Theta} \hspace{-0.5bp} {\mathbb{E}}\{\| \Delta^i_f(\bm{\theta}{,} \mu_R{,} \bm{U}_R{,} \mathbf{H})\bm{U}_R \|^2_2\} 
   &
   \hspace{-1.5bp}\leq\hspace{-1bp}
   \mathbb{E}\{(L^i_f(\mathbf{H}))^2\}\hspace{-.5bp}(N_R {+} 4)^2{.}
   \hspace{-4bp}
\end{align}
Additionally, it is true that
\begin{align}
\label{eq:21}
& \hspace{-2bp}\mathbb{E}\{ \Delta^i_f(\bm{\theta} {,} \mu_R {,} \bm{U}_R {,} \mathbf{H})\bm{U}_R{\mid} \mathbf{H}\}  \nonumber \\ 
& \hspace{8bp}\equiv \nabla_{\bm{a}} \mathbb{E}\{f^i(\Pi_\mathcal{A}\{\bm{a}{+} \mu_R  \bm{U}_R \}{,} \mathbf{H}){\mid} \mathbf{H}\}{\big\rvert}_{\bm{a} {=} \bm{\phi}(\mathbf{H},\bm{\theta})}{.}
\end{align}
\end{lemma}

Lemmata \ref{lemma2} and \ref{lemma3} are similar to (\hspace{-0.1pt}\cite{kalogerias2020model}, Lemmata 3 and 4, respectively). However, in Lemma \ref{lemma3}, owing to condition \textbf{C2} of Assumption \ref{assumption1}, there is a key difference in that we are concerned with zeroth-order gradient evaluations on $\mathcal{A}$ rather than $\Theta$. Thus, the dimension of policy parameter $N_{\bm{\phi}}$ in the respective result in  \cite{kalogerias2020model} is replaced by the action space dimension $N_R$, where $N_{\boldsymbol{\phi}} \gg N_R$.

We now present a deterministic policy gradient theorem which exploits our two-point zeroth-order gradient evaluations in action space, which will be used in the proposed primal-dual learning algorithm presented in Section \ref{sec:algo}.
\begin{theorem}[\textbf{A Deterministic Policy Gradient Theorem}]
\label{theorem:4}
\sloppy Let conditions \textbf{C2} and \textbf{C3} of Assumption \ref{assumption1} be in effect. Then, for every $\mu_R > 0$, for all $\bm{\theta} \in \Theta $, and for all $i \in \mathbb{N}^+_{N_S}$, each $\bar{f}^{\bm{\phi},i}_{\mu_R}(\bm{\theta})$ is such that
\begin{equation}
    \label{eq:22}
\nabla_{\bm{\theta}} \bar{f}^{\bm{\phi},i}_{\mu_R}(\bm{\theta}){\equiv}
\mathbb{E} \big\{ {(\Delta^i_f(\bm{\theta}, \mu_R,\bm{U}_R,\mathbf{H})\bm{U}_R)}^T \nabla_{\bm{\theta}}\bm{\phi}(\mathbf{H},\bm{\theta})\big\}.
\end{equation}
\end{theorem}
\begin{proof}
For $i \in \mathbb{N}^+_{N_S}$, consider the expectation function $\bar{f}^{\bm{\phi},i}_{\mu_R}(\cdot) \equiv \mathbb{E}\left\{f^i(\Pi_{\mathcal{A}}\{\bm{\phi}(\mathbf{H}, \cdot){+}\mu_R \bm{U}_R\}, \mathbf{H})\right\}$. Then, from the law of total expectation, the Leibniz integral rule, and \eqref{eq:21}, we verify \eqref{eq:22} as
\begin{flalign} \label{eq:27}
    &\hspace{-4bp}\nabla_{\bm{\theta}} \bar{f}^{\bm{\phi},i}_{\mu_R}(\bm{\theta})  \nonumber\\
    &{=} \nabla_{\bm{\theta}} \mathbb{E}\{ \mathbb{E}\{f^i(\Pi_{\mathcal{A}}\{\bm{\phi}(\mathbf{H},\bm{\theta}){+}\mu_R \bm{U}_R\}, \mathbf{H}) {\mid} \mathbf{H}\} \} \nonumber\\
    &{\equiv} \mathbb{E}\{ \nabla_{\bm{\theta}} \mathbb{E}\{f^i(\Pi_{\mathcal{A}}\{\bm{\phi}(\mathbf{H},\bm{\theta}){+}\mu_R \bm{U}_R\}, \mathbf{H}) {\mid} \mathbf{H}\} \} \nonumber\\
    &{\equiv} \mathbb{E}\big\{\nabla_{\bm{a}}\mathbb{E}\{f^i(\Pi_{\mathcal{A}}\{\bm{a}{+}\mu_R \bm{U}_R\}, \mathbf{H}) {\mid} \mathbf{H}\} ^T \big\rvert_{\bm{a}{=}\bm{\phi}(\mathbf{H}, \bm{\theta})}  \nonumber \\ 
    & \hspace{2in} \times\nabla_{\bm{\theta}} \bm{\phi}(\mathbf{H}, \bm{\theta}) \big\} \nonumber\\
    &{\equiv} {\mathbb{E}} \big\{ {(\Delta^i_f(\bm{\theta}{,}\mu_R{,}\bm{U}_R{,}\mathbf{H})\bm{U}_R)}^T \nabla_{\bm{\theta}}\bm{\phi}(\mathbf{H}{,}\bm{\theta})\big\},
\end{flalign}
for all $\bm{\theta} \in \Theta $.
\end{proof}

\subsection{Feasible Solutions with Surrogate}
\label{subsection:feasibility}
Before proceeding with exploiting our smoothed surrogate \eqref{eq:10} together with Theorem \ref{theorem:4}, we put forth conditions that ensure feasibility of the  surrogate, specifically on the feasible set of the \textit{original} parameterized problem \eqref{eq:1}, which is the one we can initially specify.

The premise of our formulation is to show the existence of at least one strictly feasible point for \eqref{eq:1} and \eqref{eq:10} \textit{simultaneously}. Moreover, as we also show, \eqref{eq:10} can be made strictly feasible at will (due to the feasibility slack $\texttt{S}(\mu_R)$). We first define Lipschitz constant vectors
\begin{equation} \label{eq:32}
    \bm{c}_S \, {\triangleq} \left[ L^1_g \, {\ldots} \, L^{N_\mathbf{g}}_g \right]^T \text{and} \:\, \bm{c}_R\, {\triangleq} \left[ \mathbb{E}\{L^1_f\} \, {\ldots} \, \mathbb{E}\{L^{N_S}_f\} \right]^T{,}
\end{equation}
and consider the following assumption.
\begin{assumption} \label{assumption2} The \textit{feasibility slack}  $\texttt{S}_{\mathbf{f}}$ is increasing around the origin, and $\lim_{\mu_R \to 0}\texttt{S}_\mathbf{f}(\mu_R) \equiv \texttt{S}_\mathbf{f}(0) \equiv \mathbf{0}$.
\end{assumption}
We now state results which guarantee feasibility of \eqref{eq:10} relative to \eqref{eq:1}. We analyze both the strict feasibility of \eqref{eq:10} and constraint violations in \eqref{eq:1} for every feasible solution of \eqref{eq:10}.
\begin{theorem}
\label{theorem:5}
Let Assumptions \ref{assumption1} and \ref{assumption2} be in effect, and let $(\mathbf{x}^*, \bm{\theta}^*) \in \mathbb{R}^{N_S} {\times} \mathbb{R}^{N_{\bm{\phi}}}$ be a strictly feasible point of the problem \eqref{eq:1}. Then, there exist $\mu_S^* > 0$ and $\mu_R^* > 0$, possibly dependent on $(\mathbf{x}^*, \bm{\theta}^*)$, such that, for every $0\leq\mu_S\leq\mu_S^*$ and $0\leq\mu_R\leq\mu_R^*$, the point $(\mathbf{x}^*, \bm{\theta}^*)$ is strictly feasible for  \eqref{eq:10}.
\end{theorem}
\begin{proof}
    Similar to that of Theorem 6 in \cite{kalogerias2020model}, with key differences of $N_R$ being in place of $N_{\bm{\phi}}$ and the Lipschitz vector $\bm{c}_R$ being on $\mathcal{A}$ rather than $\Theta$.
\end{proof}
\begin{theorem}
\label{theorem:6}
Let Assumption \ref{assumption1} be in effect. Then, for every $\mu_R \geq 0$ such that
\begin{align} \label{eq:38}
    \texttt{S}(\mu_R) - \mu_R \bm{c}_R \sqrt{N_R} \geq \mathbf{0}
\end{align}
and for every $\mu_S \geq 0$, every feasible point of \eqref{eq:10} is also feasible for \eqref{eq:1}. Otherwise, negative LHS values in \eqref{eq:38} correspond to the respective levels of maximal constraint violations for \eqref{eq:1}.
\end{theorem}
\begin{proof}
Similar to that of Theorem 7 in \cite{kalogerias2020model}, with key differences of $N_R$ being in place of $N_{\bm{\phi}}$ and the Lipschitz vector $\bm{c}_R$ being on $\mathcal{A}$ rather than $\Theta$.
\end{proof}

From the statements of Theorems \ref{theorem:5} and \ref{theorem:6}, we see that both of these can hold simultaneously. These, however, are different from the respective results in \cite{kalogerias2020model}, since that the slack $\texttt{S}(\mu_R)$ deals with perturbations in the action space $\mathcal{A}$ rather than the parameter space $\Theta$, which is much more flexible in terms of implementation. 

\section{Model-Free Primal-Dual Learning in Resource Allocation Spaces}
\label{sec:algo}
To efficiently tackle constrained problem \eqref{eq:1} within a fully model-free setting, the idea is to exploit the zeroth-order representations of the smoothed versions of the functions involved in \eqref{eq:1}, as introduced in Section \ref{Zeroth}.
To do this, we first define the \textit{Lagrangian} of the smoothed surrogate \eqref{eq:10} as
\begin{align} \label{eq:42}
    \mathcal{L}_{\phi,\mu}(\bm{\theta},\mathbf{x},\bm{\lambda}_S,\bm{\lambda}_R)\triangleq\, &g^o_{\mu_S}(\mathbf{x}) + \bm{\lambda}_S^T \mathbf{g}_{\mu_S}(\mathbf{x}) \nonumber \\ 
    & + \bm{\lambda}_R^T [\bar{\mathbf{f}}^{\bm{\phi}}_{\mu_R}(\bm{\theta}) - \mathbf{x} - \texttt{S}(\mu_R)].
\end{align}
Then, driven by both classical and state-of-the-art approaches, see, e.g., \cite{ribeiro2012optimal, eisen2019learning, kalogerias2020model, wang2020drl}, we are interested in the minimax problem
\begin{align}
    \begin{aligned} \label{eq:43}
    & \underset{\bm{\lambda}_S,\bm{\lambda}_R}{\text{minimize}}  && \underset{(\mathbf{x},\bm{\theta}) \in \mathcal{X} \times \Theta}{\text{sup}} \mathcal{L}_{\phi,\mu}(\mathbf{x},\bm{\theta},\bm{\lambda}_S,\bm{\lambda}_R)  \\
    & \text{subject to} && \bm{\lambda}_S \geq \mathbf{0},\bm{\lambda}_R \geq \mathbf{0}
\end{aligned}.
\end{align}
\SetInd{0.25em}{0em}
\begin{algorithm}[t]
\caption{PD-ZDPG+}
\label{algorithm1}
\SetAlgoNoLine
Initialize $\mathbf{x}^0$, $\bm{\theta}^0$, $\bm{\lambda}_S^0$, $\bm{\lambda}_R^0$, $\alpha_{\mathbf{x}}^0$, $\alpha_{\bm{\theta}}^0$, $\alpha_{\bm{\lambda}_S}^0$, $\alpha_{\bm{\lambda}_R}^0$,  $\mu_S$, $\mu_R$. \\
\For{$k=0,N{-}1$}{
    Sample $\bm{U}_S^{k+1}$ and $\bm{U}_R^{k+1}$, and measure $\mathbf{H}^{k+1}$.\\
    Evaluate $g^o(\mathbf{x}^k)$, $\mathbf{g}(\mathbf{x}^k)$, $g^o(\Pi_{\mathcal{X}}\{\mathbf{x}^k {+} \mu_S \bm{U}_S^{k{+}1}\})$ and $\mathbf{g}(\Pi_{\mathcal{X}}\{\mathbf{x}^k {+} \mu_S \bm{U}_S^{k+1}\})$.\\
    Probe the system for  $\mathbf{f}(\bm{\phi}(\mathbf{H}^{k+1},\bm{\theta}^k), \mathbf{H}^{k+1})$ and $\mathbf{f}(\Pi_{\mathcal{A}}\{\bm{\phi}(\mathbf{H}^{k+1},\bm{\theta}^k) {+} \mu_R\bm{U}_R^{k+1}\}, \mathbf{H}^{k+1})$.\\
    Update $\mathbf{x}^{k+1}$ and $\bm{\theta}^{k+1}$ using \eqref{eq:48} and \eqref{eq:49}.\\
    Evaluate  $\mathbf{g}(\Pi_{\mathcal{X}}\{\mathbf{x}^{k{+}1} {+} \mu_S \bm{U}_S^{k+1}\})$ and probe the system for $\mathbf{f}(\Pi_{\mathcal{A}}\{\bm{\phi}(\mathbf{H}^{k{+}1}{,}\bm{\theta}^{k{+}1}) {+} \mu_R\bm{U}_R^{k{+}1}\}, \mathbf{H}^{k{+}1})$.\\
    Update $\bm{\lambda}_S^{k+1}$ and $\bm{\lambda}_R^{k+1}$ using \eqref{eq:50} and \eqref{eq:51}.
}
\end{algorithm}
We now develop our proposed zeroth-order primal-dual learning algorithm to solve \eqref{eq:43}. By Lemma \ref{lemma2} and Theorem \ref{theorem:4}, gradients of $g^o_{\mu_S}$, $\mathbf{g}_{\mu_S}$ and $\bar{\mathbf{f}}^{\bm{\phi}}_{\mu_R}$ are given by expectation functions in \eqref{eq:18} and \eqref{eq:22}. Given standard Gaussian i.i.d. sequences $\{\bm{U}_S^k\}_{k\in\mathbb{N}^+}$, $\{\bm{U}_R^k\}_{k\in\mathbb{N}^+}$, and a mutually independent channel state observation sequence $\{\mathbf{H}^k\}_{k\in\mathbb{N}^+}$, we define parameter update rules employing stochastic approximation as
\SetInd{0.25em}{0.1em}
\begin{align}
    \label{eq:48}
    \mathbf{x}^{k+1} & {\equiv} {\Pi}_\mathcal{X} \{ \mathbf{x}^{k} {+} \alpha_{\mathbf{x}}^k {\circ} (\Delta^i_g(\mathbf{x}^k, \mu_S, \bm{U}_S^{k+1})\bm{U}_S^{k+1}  \nonumber \\
    &  \quad\quad\quad\quad {+} \bm{U}_S^{k+1} \bm{\Delta}_\mathbf{g}(\mathbf{x}^k, \mu_S, \bm{U}_S^{k+1})^T \bm{\lambda}^k_S {-}\bm{\lambda}^k_R) \}, \\
    \label{eq:49}
    \bm{\theta}^{k+1} & {\equiv} {\Pi}_\Theta \{ \bm{\theta}^{k} {+} \alpha_{\bm{\theta}}^k {\circ} (  \nabla_{\bm{\theta}}\bm{\phi}(\mathbf{H}^{k+1},\bm{\theta}^k)^T \bm{U}_R^{k+1}   \nonumber\\
    &  \hspace{0.75in} \times\bm{\Delta}_\mathbf{f}(\bm{\theta}^k{,}\mu_R{,}\bm{U}_R^{k+1}{,}\mathbf{H}^{k+1})^T \bm{\lambda}^k_R )\}, \\
    \label{eq:50}
    \bm{\lambda}^{k+1}_S & {\equiv} [\bm{\lambda}_S^{k} {-} \alpha_{\bm{\lambda}_S}^k{\circ}\mathbf{g}(\Pi_{\mathcal{X}}\{\mathbf{x}^{k+1} {+} \mu_S \bm{U}_S^{k+1}\})]_+ \:\: \text{and}\\
    \label{eq:51}
    \bm{\lambda}^{k+1}_R & {\equiv} [\bm{\lambda}_R^{k} {-} \alpha_{\bm{\lambda}_R}^k{\circ}(\mathbf{f}(\Pi_{\mathcal{A}}\{\bm{\phi}(\mathbf{H}^{k+1},\bm{\theta}^{k+1}) {+} \mu_R\bm{U}_R^{k+1}\}, \mathbf{H}^{k+1}) \nonumber\\ 
    & \hspace{1.65in} {-} \mathbf{x}^{k+1} {-} \texttt{S}(\mu_R))]_+,
\end{align}
where, dropping dependencies, the vectors of finite differences $\bm{\Delta}_\mathbf{g} \in \mathbb{R}^{N_\mathbf{g}}$ and $\bm{\Delta}_\mathbf{f} \in \mathbb{R}^{N_{S}}$ are defined as
\begin{align} \label{eq:52}
    \bm{\Delta}_\mathbf{g} \triangleq \left[ \Delta^1_g \ldots \Delta^{N_\mathbf{g}}_g \right]^T \quad \text{and} \quad \bm{\Delta}_\mathbf{f} \triangleq \left[ \Delta^1_f \ldots \Delta^{N_{S}}_f \right]^T
\end{align}
respectively. A complete description of our model-free primal-dual method is provided in Algorithm \ref{algorithm1}.

As problem \eqref{eq:10} is non-convex (due to $\bar{\mathbf{f}}^\phi_{\mu_R}(\bm{\theta})$), it is difficult to guarantee convergence of Algorithm \ref{algorithm1} theoretically. Nevertheless, it is true that if the algorithm converges, then it discovers a feasible (though possibly suboptimal) solution of the problem. In the next section, we evaluate the performance of Algorithm \ref{algorithm1} empirically instead, on two indicative examples.
\begin{figure}[t]
  \centering
  \centerline{\includegraphics[width=3.4in]{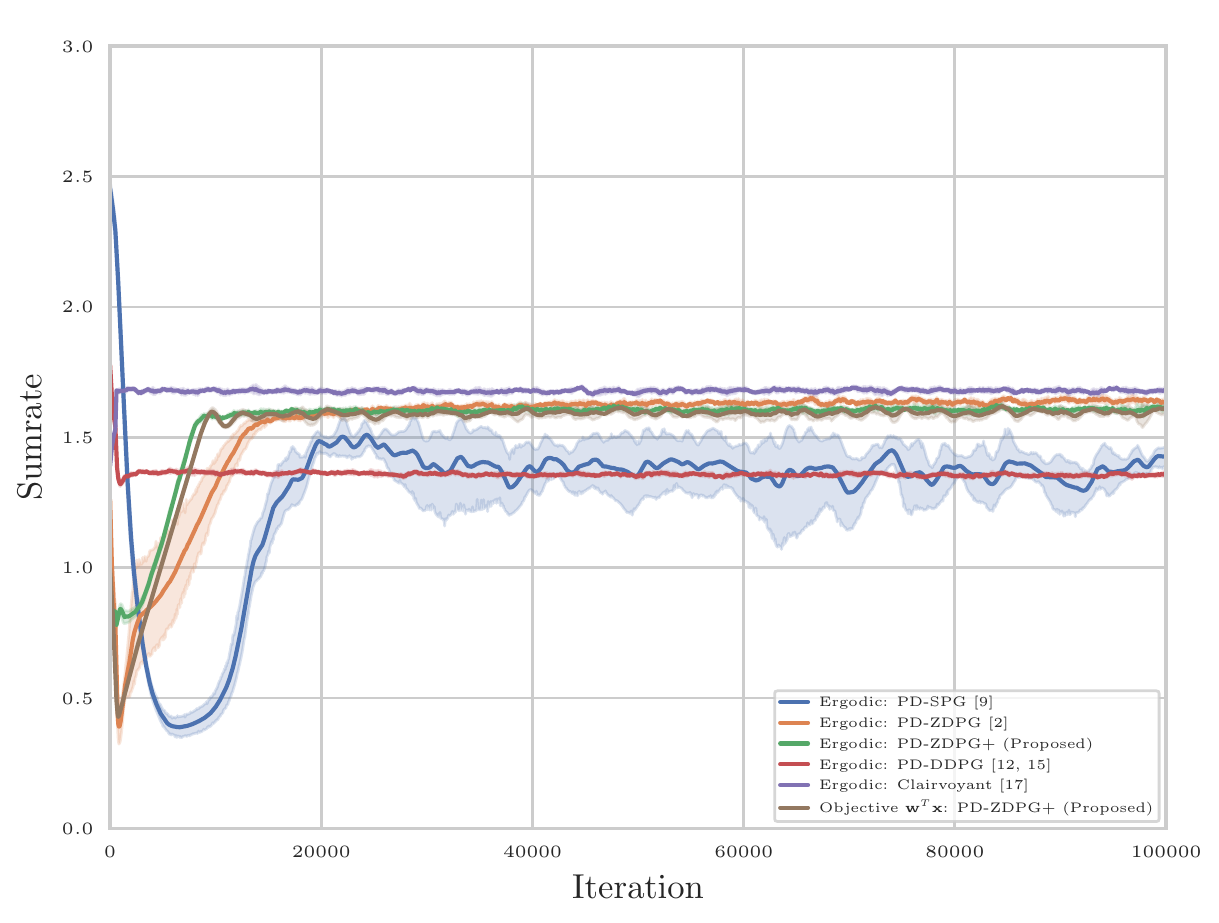}}
\caption{Sumrates (\textit{in nats per unit of time}) achieved by the proposed PD-ZDPG+, PD-SPG \cite{eisen2019learning}, PD-ZDPG \cite{kalogerias2020model}, PD-DDPG \cite{wang2020drl, lillicrap2015continuous} and the clairvoyant policy \cite{wang2010stochastic} for the AWGN channel case.}
\label{fig:awgn_1}
\end{figure}

\vspace{-4bp}
\section{Simulations}
\label{sec:sim}
\vspace{-4bp}

We consider two basic wireless models, an Additive White Gaussian Noise (AWGN) channel, and a Multiple Access Inference (MAI) channel. All experiments that follow were rerun five times and plotted in terms of mean performances and bootstrap confidence bands without any \textit{cherry-picking}. 

For the AWGN channel case, we consider a scenario where multiple users are given dedicated channels to communicate, without interference. The goal is to achieve optimal power allocation given an average total power budget $p_{max}$. This problem can be formally stated as
\begin{align} 
\begin{aligned} \label{eq:53}
    \underset{x^i,\bm{\theta}^i}{\text{maximize}} & \quad \sum_{i} w^i x^i \\
    \text{subject to} & \quad x^i \leq \mathbb{E}\left\{\log\left( 1 + \frac{H^i\phi^i(H^i, \bm{\theta}^i)}{v^i}\right)\right\}  \\
    & \quad \mathbb{E}\left\{ \sum_{i} \phi^i(H^i, \bm{\theta}^i) \right\} \leq p_{max}  \\
    & \quad (x^i, \bm{\theta}^i) \in \mathbb{R}_+ \times \mathbb{R}^{N_\phi^i}, \forall i\in {\mathbb{N}_{N_S}^+}
\end{aligned},
\end{align}
where all user weights $w^i$ are positive, randomly generated, and sum to 1, and where the policy is uncoupled among users. This is due to the simplicity of problem \eqref{eq:53}, for which a specially-structured strictly optimal \textit{clairvoyant} solution requiring complete system model information is available \cite{wang2010stochastic}; this also defines an ultimate performance benchmark under the AWGN setup.

We consider a 10-user case for the AWGN channel problem and set $p_{max}\equiv20$. Also, as in all experiments in this section, noise variance $v^i {\equiv} 1$ and $H^i$ is exponentially distributed with parameter $\lambda {\equiv} 1/2$. We then compare the proposed method (PD-ZDPG+) with other primal-dual methods including the stochastic (actor-only) policy gradient (PG) method of \cite{eisen2019learning} (PD-SPG), the deterministic zeroth order actor-only PG method of \cite{kalogerias2020model} (PD-ZDPG), and a customized primal-dual version of a deep deterministic actor-critic policy gradient method \cite{lillicrap2015continuous} (PD-DDPG). 


\begin{figure}[t]
  \centering
  \centerline{\includegraphics[width=3.5in]{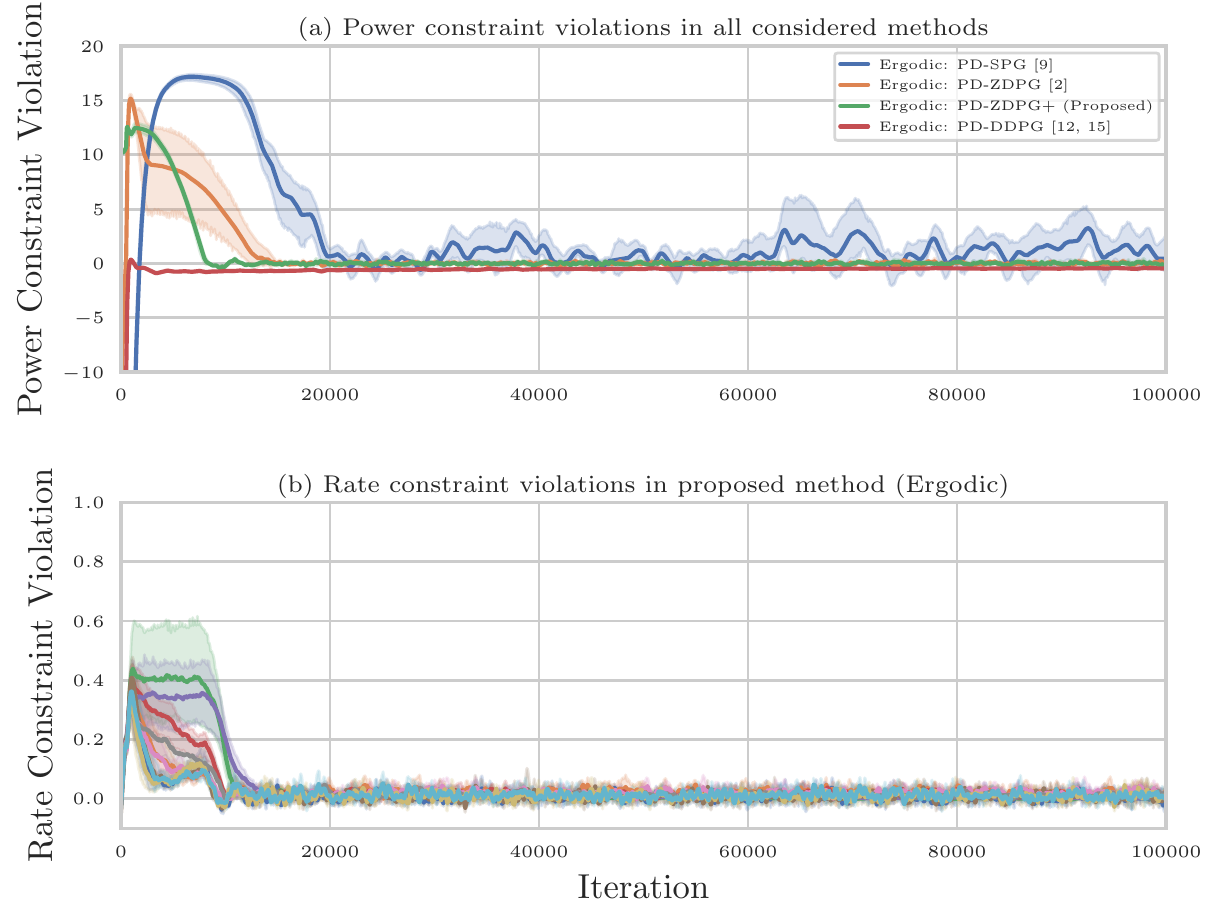}}
\caption{(a) Power constraint violation exhibited by the proposed PD-ZDPG+, PD-SPG \cite{eisen2019learning}, PD-ZDPG \cite{kalogerias2020model} and PD-DDPG \cite{wang2020drl, lillicrap2015continuous}, for the AWGN channel case.(b) Rate constraint violations of PD-ZDPG+.}
\label{fig:awgn_2}
\end{figure}

In all deterministic policy gradient methods, each policy features a ReLU-activated three-layer single-input single-$p_{max}$-sigmoid-output feed-forward DNN with eight and four neurons in the hidden layers (i.e., ten $1{-}8{-}4{-}1$ DNNs). For the stochastic policies, we define the same networks, but with two sigmoid outputs scaled by $p_{max}$ and $\sqrt{p_{max}}$, corresponding to the
mean and variance of a truncated normal distribution from which actions are sampled \cite{eisen2019learning}. For the global \textit{critic} in PD-DDPG, we consider a ReLU-activated three-layer feed-forward DNN with twenty and forty neurons in the hidden layers with all parameters initialized at $10^{-6}$. The inputs to the \textit{critic} are instantaneous channel values, whereas all policy action values are concatenated with the output of the first hidden as input to the second hidden layers. All actor parameters $\bm{\theta}$ are initialized to $0$'s while all metrics $\mathbf{x}$ and dual variables are initialized to $1$'s.

Given that the objective description is known, we smooth only the constraints of the problem (cf. \eqref{eq:10}). The exact learning rates for all the methods can be found in Table \ref{table:awgn_1}. For PD-SPG, we use the Adam optimizer with $\beta_1$ and $\beta_2$ set to $0.9$ and $0.999$ respectively with batch size of 32 where as all other methods use SGD optimization. We show the convergence of all methods over $10^5$ iterations in Figure \ref{fig:awgn_1} and report that the proposed method converges both faster and to a near-optimal solution as compared with other methods. We also show a similar behavior for constraint violations in Figure \ref{fig:awgn_2}.

\begin{table}[ht]
\begin{center}
\begin{tabular}{ p{1.18in}|p{0.25in}p{0.6in}p{0.25in}p{0.25in}}
 \hline
 Method & $\alpha_{\mathbf{x}}^k$ & $\alpha_{\bm{\theta}}^k$ & ${\alpha_{\bm{\lambda}_R}^k}$ & $\alpha_{\bm{\lambda}_R}^k*$\\
 \hline
 PD-SPG \cite{eisen2019learning} &0.01&0.01&0.0001&0.08\\
 PD-ZDPG \cite{kalogerias2020model} &0.001&0.0008&0.008&0.0001\\
 PD-ZDPG+ \hspace{-0.02in}(Proposed) &0.001&0.02&0.008&0.0001\\
 PD-DDPG \cite{wang2020drl, lillicrap2015continuous} &0.001&0.002/0.001&0.01&0.0001\\
 \hline
\end{tabular}
\end{center}
\caption{Learning rates for the AWGN channel ($*$ Power constraint)}
\label{table:awgn_1}
\end{table}

Optimizing stochastic policies proves to be a noisy procedure. In fact, these policies not only converge to suboptimal solutions, but are also not desirable as they inject extra noise into the system during implementation, inducing unnecessary statistical variability in performance.  Interestingly, PD-DDPG also converges to the same solution as PD-SPG showing that using a \textit{critic} is indeed a heuristic. As such, we witness the credit assignment problem in terms of the individual policies becoming apparent in PD-DDPG as the critic inaccurately approximates the functions in the constraints. PD-ZDPG+, on the other hand, approaches the globally optimal solution to the AWGN problem at hand rather faithfully, achieving similar performance to PD-ZDPG, which relies on high-dimensional parameter space exploration.

\begin{figure}[t]
  \centering
  \centerline{\includegraphics[width=3.5in]{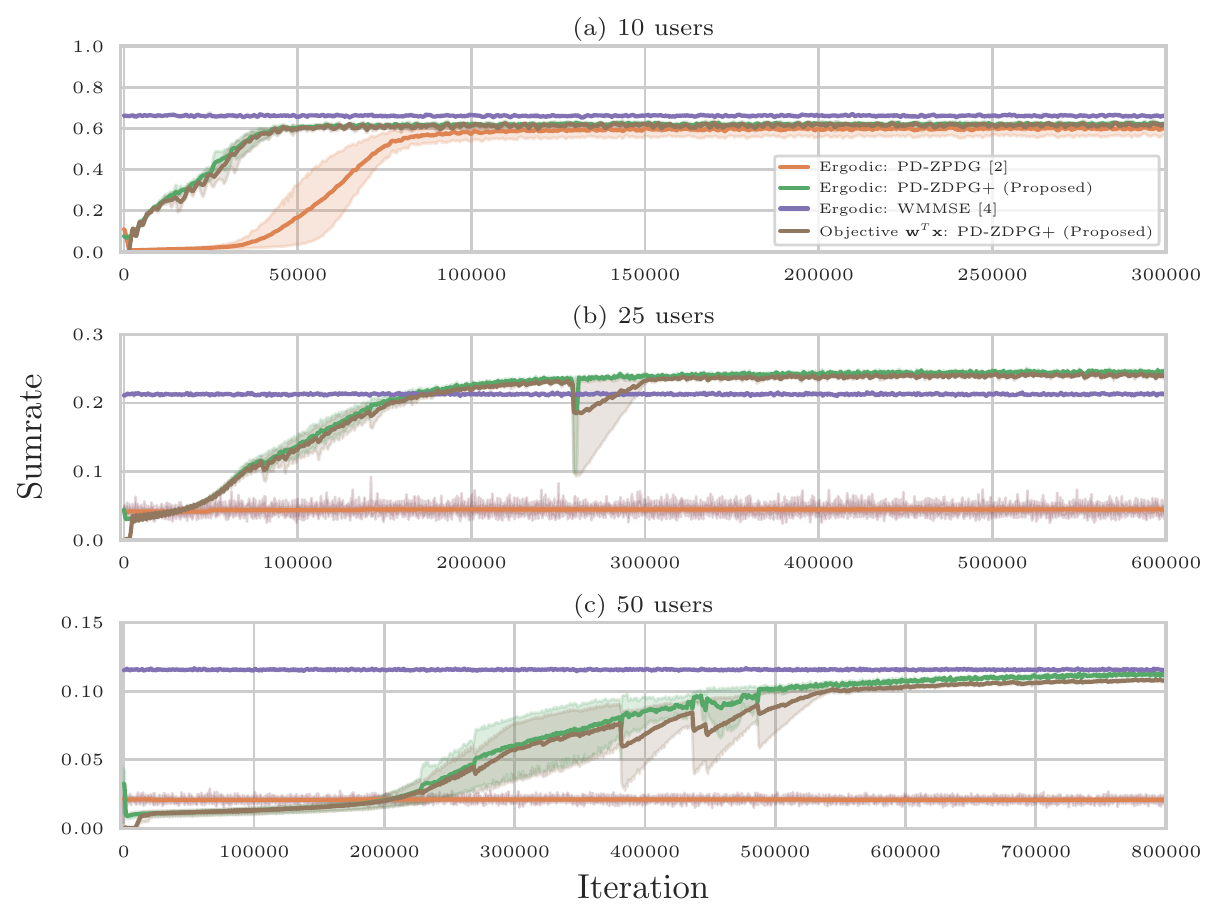}}
\caption{Sumrates (\textit{in nats per unit of time}) achieved by the proposed PD-ZDPG+, PD-ZDPG \cite{kalogerias2020model} and the WMMSE policy \cite{shi2011iteratively} for the MAI channel case with (a) 10 users (b) 25 users and (c) 50 users.}
\label{fig:mai_1}
\end{figure}

Comparable performances of PD-ZDPG \cite{kalogerias2020model} and our PD-ZDPG+ on the simple AWGN channel problem shows that our method performs as well as the state-of-the-art. To further compare the scalability of these algorithms, we consider a MAI channel model setting, in which $N_S$ transmitters communicate simultaneously with a central node. Thus, a signal transmitted by each user will interfere with the signals transmitted by all the other users. Again, the task is to optimize power allocation per user given an average power budget of $p_{max}$. In this case, our resource allocation problem is formulated as  
\begin{align} 
\begin{aligned} \label{eq:54}
\noindent \underset{x^i,\bm{\theta}}{\text{maximize}} & \:\sum_{i} w^i x^i \hspace{1in}\\
    \text{subject to} & \: x^i {\leq} \mathbb{E}\left\{{\log}\left( 1 {+} \frac{H^i\phi^i(H{,} \bm{\theta})}{v^i {+} \sum_{i{\neq}j} H^j\phi^j(H{,} \bm{\theta})}\right)\right\} \\
    & \: \mathbb{E}\left\{ \sum_{i} \phi^i(H, \bm{\theta}) \right\} \leq p_{max}  \\
    & \: (x^i, \bm{\theta}) \in \mathbb{R}_+ \times \mathbb{R}^{N_\phi}, \forall i\in {\mathbb{N}_{N_S}^+}
\end{aligned}.
\end{align}

In the following, we consider the MAI problem for 10, 25 and 50 users, where, again, the total allocated power budget $p_{max} \equiv 20$ and all the user weights $w^i$ are positive, randomly generated, and sum to $1$. For our implementations, we define global policies via ReLU activated three layer neural networks having 64 and 32 neurons in the hidden layers respectively (a single $\#{-}64{-}32{-}\#$ DNN, where ``$\#$'' is the number of users). All policy parameters $\bm{\theta}$ and all metrics $\mathbf{x}$ are initialized to $0$'s while all dual variables are initialized to $1$'s.

\begin{table}[ht]
\begin{center}
\begin{tabular}{ p{1.18in}|p{0.25in}p{0.45in}p{0.25in}p{0.25in}}
 \hline
 Method & $\alpha_{\mathbf{x}}^k$ & $\alpha_{\bm{\theta}}^k$ & ${\alpha_{\bm{\lambda}_R}^k}$ & $\alpha_{\bm{\lambda}_R}^k*$\\
 \hline
 PD-ZDPG \cite{kalogerias2020model} &0.001&0.00005&0.004&0.0001\\
 PD-ZDPG+ \hspace{-0.02in}(Proposed) &0.001&0.04&0.008&0.0001\\
 \hline
\end{tabular}
\end{center}
\caption{Learning rates for the MAI channel ($*$ Power constraint)}
\label{table:mai_1}
\end{table}

We compare these methods with the well-known WMMSE policy \cite{shi2011iteratively} (at saturation) as a benchmark upper bound. We use the same parameterization and learning rates in all experiments as given in Table \ref{table:mai_1}. PD-ZDPG consistently converged to the optimal solution in all cases (even occasionally outperforming WMMSE), whereas for the 25- and 50-user cases, PD-ZPDG did not converge and we were unable to find effective learning rates to make the method convergent. Finally, we observed that PD-ZDPG+ is also much faster in terms of training times (in Python), as shown in Figure \ref{fig:mai_2}. These comparisons are justified according to the documentation of Python's PyTorch library, which we have used in our experiments.

\section{Conclusion}
\label{sec:conclude}

In this work, we have proposed a new primal-dual method for learning resource allocations in wireless systems by exploiting zeroth-order deterministic policy gradients via low-dimensional action space exploration. Our method works with powerful policy parameterizations such as DNNs, and outperforms the current state-of-the-art, both in terms of optimal convergence and scalability. This is due to the fact the our method optimizes deterministic policies, is completely model-free, and is not limited, in terms of scalability, by the dimensionality of employed policy parameterizations. Although we have only tested feed-forward neural networks in our numerics, PD-ZDPG+ works equally fine with other deep learning models like convolutional and recurrent neural networks, as well as architectures which exploit intrinsic network structure, such as random edge graph neural networks. Therefore, PD-ZDPG+ can be applied to a wide range of constrained resource allocation problems.  In the future, we would like to evaluate the behavior of PD-ZDPG+ for such more elaborate policy parameterizations. We are also interested to see how our proposed method extends and performs on problems which are outside the scope of wireless systems resource allocation, such as problems arising in general purpose constrained RL. 
\begin{figure}[H]
  \centering
  \centerline{\includegraphics[width=3.5in]{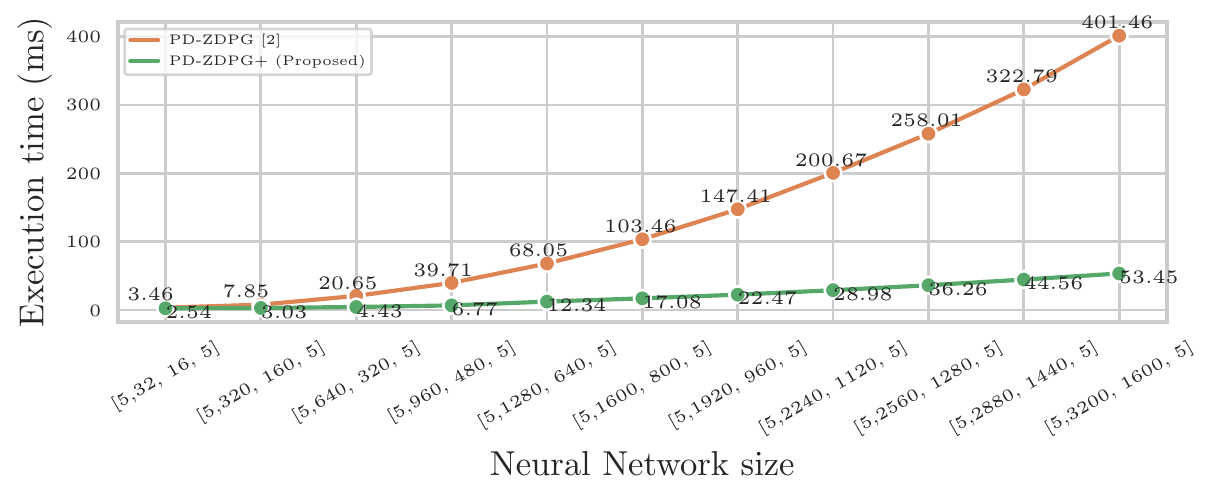}}
\caption{Best training execution times per iteration (in ms) of the proposed PD-ZDPG+ and PD-ZDPG \cite{kalogerias2020model} for the MAI channel case with 5 users and multiple neural network sizes, using 2 virtual Intel Xeon CPUs @ 2.30GHz, 13GB of RAM, and an 8 core TPU v3.}
\label{fig:mai_2}
\end{figure}

\bibliographystyle{IEEEbib}
\bibliography{PD-ZDPG+}
\end{document}